\documentclass[copyright,creativecommons]{eptcs}
\providecommand{\event}{FIT 2012}

\usepackage{epsfig}
\usepackage{graphicx}
\usepackage{boxedminipage}
\usepackage{multicol}
\usepackage[usenames]{color}
\usepackage{amsmath,amssymb,latexsym,gastex}
\usepackage{wrapfig}
\newcommand{\tuple}[1]{\langle #1 \rangle}
\newcommand{\ignore}[1]{}
\newcommand{\pc}{\textit{pc}}

\newcommand{\atomic}{\mbox{\textit{atom}}}

\newcommand{\zug}[1]{\langle #1 \rangle}
\title{
Sequentializing Parameterized Programs
\thanks{This research was supported in part by 
MIUR-FARB 2009-2011 
Universit\`a degli Studi di Salerno and
NSF Career Award 0747041.}
}

\newtheorem{theorem}{Theorem}
\newtheorem{lemma}{Lemma}

\newtheorem{definition}{Definition}
\newtheorem{example}{Example}
\newcommand{\qed}{\hfill $\Box$}

\author{
Salvatore {La Torre}
\institute{
Dipartimento di Informatica\\
Universit\`a degli Studi di Salerno, Italy}
\email{slatorre@unisa.it}
\and
P. Madhusudan
\institute{
Department of Computer Science\\
University of Illinois at Urbana-Champaign, USA
}
\email{madhu@cs.illinois.edu}
\and 
Gennaro Parlato
\institute{
Department of Electronic \& Computer Systems\\
University of Southampton, UK}
\email{gennaro@ecs.soton.ac.uk}
}

\def\titlerunning{Sequentializing Parameterized Programs}
\def\authorrunning{S. {La Torre}, P. Madhusudan \& G. Parlato}

\begin{document}
\maketitle

\begin{abstract}
We exhibit assertion-preserving (reachability preserving) transformations from parameterized concurrent
shared-memory programs, under a $k$-round scheduling of processes, to sequential programs.
The salient feature of the sequential program is that it tracks the local variables of only one thread at any point, and uses only $O(k)$ copies of shared variables (it does not use extra counters, not even one counter to keep track of the number of threads). Sequentialization is achieved using the concept of a
linear interface that captures the effect \emph{an unbounded block} of processes have on the shared state in a $k$-round
 schedule. 
Our transformation utilizes linear interfaces to sequentialize the program, and to ensure the
 sequential program explores only reachable states and preserves local invariants.
\end{abstract}



\section{Introduction}


%

The theme of this paper is to build verification techniques for parameterized concurrent shared-memory
programs: programs with \emph{unboundedly} many threads, many of them running identical code that concurrently
evolve and interact through shared variables.
Parameterized concurrent programs are extremely hard to check for errors. Concurrent shared-memory programs
with a finite number of threads are already hard, and extending existing methods for sequential
programs, like Floyd-Hoare style deductive verification, abstract interpretation,
and model-checking, is challenging.
Parameterized concurrent programs are even harder.


A recent proposal in the verification community to handle concurrent program
verification is to \emph{reduce the problem to sequential program verification}.
Of course, this is not possible, in general, unless the sequential program
tracks the entire configuration of the concurrent program, including the local
state of each thread.
However, recent research has shown efficient sequentializations for
concurrent programs with finitely many threads
when restricted to a \emph{fixed number of rounds of schedule} (or a fixed
number of context-switches)~\cite{LalReps,LMPCAV09}.
A round of schedule involves scheduling each process, one at a time in some order,
where each process is allowed to take an \emph{arbitrary} number of steps.

The appeal of sequentialization is that by reducing concurrent
program verification to sequential program verification, we can bring to bear
all the techniques and tools available for the analysis of sequential programs.
Such sequentializations have been used recently to convert concurrent programs under bounded
round scheduling to sequential programs, followed by automatic deductive verification techniques based on
SMT solvers~\cite{ZvonimirCAV09}, in order to find bugs (see also~\cite{QadeerPOPL11}).
The goal of this paper is to find a similar translation for \emph{parameterized} concurrent
programs where the number of threads is unbounded. Again, we are looking for an efficient translation---
to convert a parameterized concurrent
program to a sequential program so that the latter tracks \emph{the local state of at most one thread
at any time, and uses only a bounded number of copies of shared variables}.



The motivation for the bounded round restriction is inspired from recent research in
testing that suggests that most concurrency errors manifest themselves within a few context-switches
(executions can be, however, arbitrarily long). The {\sc Chess} tool from Microsoft Research, for example,
tests concurrent programs only under schedules that have a bounded number of context-switches
(or pre-emptions)~\cite{CHESS}. In the setting where there is an unbounded number of threads,
the natural extension of bounded context-switching is bounded-round context-switching, as the latter executes schedules that allow all threads in the system to execute (each thread is context-switched into a bounded number
 of times). Checking a parameterized program to be correct for a few rounds
gives us considerable confidence in its correctness.

The main result of this paper is  that efficient sequentializations of \emph{parameterized} programs
are also feasible, when restricted to bounded-round context-switching schedules. More precisely, we show
that given a parameterized program $P$ with an unbounded number of threads executing under
$k$-round schedules and an error state $e$, there is an effectively constructible (non-deterministic)
sequential program $S$ with a corresponding error state $e'$ that satisfies the following:
(a) the error state $e$ is reachable in $P$ iff the error state $e'$ is reachable in $S$,
(b) a localized state (the valuation of a thread's local variables
and the shared variables) is reachable in $P$ iff it is reachable in $S$--- in other words,
the transformation preserves assertion invariants and explores only reachable states,
(we call such a transformation \emph{lazy}), and (c) $S$ at any point tracks the local state of only one
thread and at most $O(k)$ copies of
shared variables, and furthermore, uses no additional unbounded memory such as counters or queues.

The existence of a transformation of the above kind is theoretically challenging and
interesting. First, when simulating a parameterized program, one would expect that \emph{counters} are
necessary--- for instance, in order to simulate the second round of schedule,
it is natural to expect that the sequential program must remember at least the number of threads it used in the
first round. However, our transformation does not introduce counters.
Second, a lazy transformation is hard as, intuitively, the sequential program needs to
\emph{return} to a thread in order to simulate it, and yet cannot keep the local state during this process.
The work reported in~\cite{LMPCAV09} achieves such a lazy sequentialization for concurrent programs with
finitely many threads using \emph{recomputation} of local states.
Intuitively, the idea is to fix a particular ordering of
threads, and to restart threads which are being context-switched into to recompute their local state.
Sequentializing parameterized concurrent programs is significantly harder as we, in addition, do not even
know how many threads were active or how they were ordered in an earlier round,
let alone know the local states of these threads.


Our main technical insight to sequentialization is to exploit a concept called
\emph{linear interface} (introduced by us recently~\cite{CAV10} to provide model-checking
algorithms for \emph{Boolean} parameterized systems).
A linear interface summarizes the effect of an \emph{unbounded block of processes}
on the shared variables in a $k$-rounds schedule.
A linear interface is of the form $(\overline{u}, \overline{v})$ where
$\overline{u}$ and $\overline{v}$ are $k$-tuples of valuations of shared variables.
A block of threads have a  linear interface $(\overline{u}, \overline{v})$
if, intuitively, there is an execution which allows context-switching into this block
with $u_i$ and context-switch out at $v_i$, for $i$ growing consecutively from $1$ to $k$,
\emph{while preserving the local state} across the context-switches
(see Figure~\ref{fig:theoLI}).

In classic verification of sequential programs with recursion (in both deductive verification
as well as model-checking), the idea of summarizing procedures using pre-post conditions
(or summaries in model-checking algorithms) is crucial in handling the infinite recursion depth.
In parameterized programs, linear interfaces have a similar role but across a concurrent dimension:
they capture the pre-post condition for \emph{a block of unboundedly many processes}.
However, because a block of processes
has a persistent state across two successive rounds of a schedule (unlike a procedure
called in a sequential program), we cannot summarize the effect of a block using a pre-post predicate
that captures a set of pairs of the form $(u,v)$. A linear interface captures a sequence of length $k$
of pre-post conditions, thus capturing in essence the effect of the local states of the block of processes
that is adequate for a $k$-round schedule.

Our sequentialization synthesizes a sequential program that uses a recursive procedure to compute linear interfaces. 
Intuitively, linear interfaces
of the parameterized program correspond to \emph{procedural summaries} in the sequential program.
Our construction hence produces a recursive program even when the parameterized program has no recursion.
However, the translation also works when the parameterized program has recursion, as the program's recursion gets
properly nested within the recursion introduced by the translation.

Our translations work for general programs, with unbounded data-domains.
When applied to parameterized programs over \emph{finite data-domains}, it shows
that a class of \emph{parameterized pushdown systems} (finite automata with
an unbounded number of stacks) working under a bounded round schedule, can be simulated faithfully
by a \emph{single-stack pushdown system} (see Section~4.2). 

The sequentialization provided in this paper paves the way to finding errors in parameterized
concurrent programs with unbounded data domains, up to a bounded number of
rounds of schedule, as it allows us to convert them to sequential programs, in order to apply
the large class of sequential analysis tools available---model-checking, testing,
verification-condition generation based on SMT solvers, and abstraction-based verification.
The recent success in finding errors in concurrent systems code (for a finite number of threads)
using a sequentialization and then SMT-solvers~\cite{ZvonimirCAV09,QadeerPOPL11} lends credence to this optimism.

\paragraph{{\bf Related work.}}

 The idea behind bounding context-switches is that most concurrency errors manifest within a few
 switches~\cite{QadeerWu,MusuvathiQadeerPLDI07}.
 The {\sc Chess} tool from Microsoft espouses this philosophy by testing concurrent programs
 by systematically choosing all schedules with a small number of preemptions~\cite{CHESS}.
 Theoretically, context-bounded analysis was motivated for the
 study of concurrent programs with bounded data-domains and recursion, as it yielded a decidable
 reachability problem~\cite{QadeerRehof}, and has been exploited in model-checking~\cite{SES08,LalReps,LaTorreMadhusudanParlatoPLDI09}.
 In recent work~\cite{CAV10}, we have designed model-checking
 algorithms for Boolean abstractions of parameterized programs using the concept of linear interfaces; these work only for bounded data domains and also do not give sequentializations.

 The first sequentialization of concurrent programs was proposed for a finite number of threads
 and two context-switches~\cite{QadeerWu}, followed
 by a general \emph{eager} conversion that worked for arbitrary number of context-switches~\cite{LalReps},
 and a \emph{lazy} conversion proposed by us in \cite{LMPCAV09}.
 Sequentialization has been used recently on concurrent device drivers written in C with
 dynamic heaps, followed by using proof-based verification techniques to find bugs~\cite{ZvonimirCAV09}.
 A sequentialization for delay-bounded schedulers that allows exploration of concurrent programs
 with dynamic thread creation has been discovered recently~\cite{QadeerPOPL11}.
 Also, in recent work, it has been established that any concurrent program with finitely many threads
 that can be reasoned with compositionally, using a rely-guarantee proof, can be sequentialized~\cite{Garg11}.

 A recent paper proposes a solution using Petri net reachability to the reachability problem
 in concurrent programs with \emph{bounded data domains} and  dynamic creation of threads, where a thread is context-switched into only a bounded number
 of times~\cite{ABQ09}. Since dynamic thread creation can model unboundedly many threads, this framework is more powerful (and much more expensive in complexity) than ours, \emph{when restricted to bounded data-domains}.

 There is a rich history of verifying parameterized asynchronously communicating
 concurrent programs, especially motivated by the verification of distributed protocols.
 We do not survey these in detail (see~\cite{CohenNamjoshiCAV08,KPSZ02,KestenMalerMarcusPnueliShahar97,EmersonKahlonCSL04,PnueliXuZuckCAV02,BaslerMazzucchiWahlKroeningCAV09},
  for a sample of this research).


\section{Preliminaries}
\label{sec:pre}

\paragraph{\bf Sequential recursive programs.}

Let us fix the syntax of a simple \emph{sequential} programming
language with variables ranging over only the integer and Boolean domains,
with explicit syntax for nondeterminism,  and (recursive) function calls.
For simplicity of exposition, we do not consider dynamically allocated structures
or domains other than integers; however, all results in this paper can be easily extended
to handle such features.

Sequential programs are described by the following grammar:
\begin{tabbing}
$\langle \textit{seq-pgm}\rangle$ ~~\= ::=~~ \= $\langle \textit{vardec}\rangle {\tt ;}~ \langle\textit{sprocedure}\rangle^*$ \\
$\langle \textit{vardec}\rangle$ \> ::= \> $\langle \textit{type}\rangle~ x \mid
\langle \textit{vardec}\rangle \mathtt{;}~ \langle \textit{vardec}\rangle$\\
$\langle \textit{type}\rangle$ \> ::= \> ${\tt int} \mid {\tt bool}$\\
$\langle\textit{sprocedure}\rangle$ \> ::= \>  ($\langle type \rangle \mid {\tt void}) f(x_1,\ldots,x_h) ~{\tt begin}~
       \langle\textit{vardec}\rangle {\tt ;}~ \langle\textit{seq-stmt}\rangle ~{\tt end}$\\
$\langle\textit{seq-stmt}\rangle$ \> ::= \>
 $\langle\textit{seq-stmt}\rangle \textit{;}~ \langle\textit{seq-stmt}\rangle \mid {\tt skip}
 \mid x := \langle\textit{expr}\rangle \mid {\tt assume}(\langle\textit{b-expr}\rangle) \mid$ \\
\> \> ${\tt call}~f(x_1, \ldots, x_h) \mid
{\tt return} ~x \mid \mathtt{while}~ (\langle\textit{b-expr}\rangle) ~\mathtt{do}~ \langle \textit{seq-stmt}\rangle ~\mathtt{od}$ \\
\> \> ${\tt if}~ (\langle\textit{b-expr}\rangle)~ {\tt then}~
\langle\textit{seq-stmt}\rangle~ {\tt else}~
\langle\textit{seq-stmt}\rangle~ {\tt fi} \mid {\tt assert} \langle\textit{b-expr}\rangle$\\
$\langle\textit{expr}\rangle$ \> ::= \> $x \mid c \mid f(y_1, \ldots, y_h) \mid\langle \textit{b-expr} \rangle$ \\
$\langle\textit{b-expr}\rangle$ \> ::= \> $\textit{T} \mid \emph{F} \mid * \mid x \mid
\neg\langle\textit{b-expr}\rangle \mid \langle\textit{b-expr}\rangle \vee \langle\textit{b-expr}\rangle$ \\
\end{tabbing}

\vspace*{-0.2truecm}
Variables are scoped in two ways, either as global variables shared between procedures,
or variables local to a procedure, according to where they are declared.
Functions are all call-by-value. Some functions $f$ may be interpreted to have
existing functionality, such as integer addition or library functions,
in which case their code is not given and we assume they happen atomically.
We assume the program is well-typed according to the type declarations.

Note that Boolean expressions can be true, false, or non-deterministically true or false ($*$),
and hence programs are non-deterministic (which will be crucial as we will need to simulate
concurrent programs, which can be non-deterministic). These non-deterministic choices
can be replaced as \emph{inputs} in a real programming language
if we need to verify the sequential program.



Let us assume that there is a function {\tt main}, which is the
function where the program starts, and that there are no calls to
this function in the code of $P$.
The semantics of a sequential program $P$ is the obvious one.

The {\tt assert} statements form the specification for the program, and express
invariants that can involve all variables in scope.
Note that \emph{reachability} of a particular statement can be encoded using
an {\tt assert F} at that statement.

\paragraph{\bf Parameterized programs with a fixed number of shared variables.}
We are interested in concurrent programs composed of several concurrent processes,
each executing on possibly unboundedly many threads (\emph{parameterized programs}).
All threads run in parallel and share a fixed number of variables.

A \emph{concurrent process} is essentially a sequential program with the possibility of declaring
sets of statements to be executed \emph{atomically}, and is given by the following grammar
(defined as an extension on the syntax for sequential programs):
\begin{tabbing}
~~$\langle\textit{process}\rangle$ ~~~~~~~~\= ::= ~~\=  ${\tt process} ~P~ {\tt begin}~ \langle\textit{vardec}\rangle\textit{;}~
\langle\textit{cprocedure}\rangle^*~  {\tt end}$ \\
~~$\langle\textit{cprocedure}\rangle$ \> ::= \>
$(\langle type \rangle \mid {\tt void}) f(x_1,\ldots,x_h) ~{\tt begin}~\langle\textit{vardec}\rangle {\tt ;} \langle\textit{conc-stmt}\rangle~ {\tt end}$\\
~~$\langle\textit{conc-stmt}\rangle$ \> ::= \>  $\langle\textit{conc-stmt}\rangle {\tt ;} \langle\textit{conc-stmt}\rangle
 \mid \langle\textit{seq-stmt}\rangle \mid {\tt atomic~~ begin}~ \langle\textit{seq-stmt}\rangle~
{\tt end}$
 \end{tabbing}

\pagebreak
\noindent The syntax for parameterized programs is obtained by adding the following rules:
\begin{tabbing}
~$\langle\textit{param-pgm}\rangle$ ~~~~~\= ::=~ \=
$\langle\textit{vardec}\rangle \langle\textit{init}\rangle \langle\textit{process}\rangle^*$\\
~$\langle\textit{init}\rangle$ \> ::= \> $\langle\textit{seq-stmt}\rangle$\\
 \end{tabbing}

\vspace*{-0.2truecm}
Variables in a parameterized program can be scoped locally, globally (i.e. to
a process at a particular thread) or shared (shared amongst all processes in
all threads, when declared before {\tt init}). The statements and assertions in a parameterized
program can refer to all variables in scope.

Each parameterized program has a sequential block of statements, {\tt init}, where
the shared variables are initialized. The parameterized program is initialized
with an arbitrary finite number of threads, each thread running a copy of one of
the processes.
Dynamic creation of threads is not allowed. However, dynamic creation can be modeled
(at the cost of a context-switch per created thread) 
by having threads created in a ``dormant'' state, which get active when they
get a message from the parent thread to get created.

An \emph{execution} of a parameterized program is obtained by interleaving
the behaviors of the threads which are involved in it.

Formally, let ${\cal P}=(S,{\tt init},\{P_i\}_{i=1}^{n})$ be a
\emph{parameterized program} where $S$ is the set of shared variables and
$P_i$ is a process for $i=1,\ldots,n$.
We assume that each statement of the program has a unique \emph{program
counter} labeling it.
A \emph{thread} $T$ of $\cal P$ is
a copy (instance) of some $P_i$ for $i=1,\ldots,n$.
At any point, only one thread is \emph{active}.
For any $m>0$, a \emph{state} of $\cal P$ is denoted by a tuple
$(map, i, s,\sigma_1,\ldots,\sigma_m)$ where:
(1) $map: [1,m] \rightarrow P$ is a mapping from
    threads $T_1, \ldots T_m$ to processes,
(2) the thread which is currently active is $T_i$, where $i \in [1,m]$
(3) $s$ is a valuation of the shared variables, and
(4) $\sigma_j$ (for each $j \in [1,m]$) is a local state of $T_j$.
Note that each $\sigma_j$ is a \emph{local state} of a process, and
is composed of a valuation of the program counter, local, and global variables of the process,
and a \emph{call-stack} of local variable valuations and program counters
to model procedure calls.

At any state $(map, i, s,\sigma_1,\ldots,\sigma_m)$, the valuation of the shared variables $s$
is referred to as the \emph{shared state}.
A \emph{localized state} is the \emph{view} of the state by the current process,
i.e. it is $(\widehat{\sigma_i}, s)$,
where $\widehat{\sigma_i}$ is the component of $\sigma_i$ that defines the valuation of local
and global variables, and the local pc (but not the call-stack),
and $s$ is the valuation of the shared variables in scope.
Note that assertions express properties of the localized state only.
Also, note that when a thread is not scheduled, the local state of its process
does not change.

The interleaved semantics of parameterized programs
is given in the obvious way.
We start with an arbitrary state, and execute the statements of {\tt init}
to prepare the initial shared state of the program, after which the threads become active.
Given a state $(map, i,\nu,\sigma_1,\ldots,\sigma_m)$, it can either fire a \emph{transition}
of the process at thread $T_i$ (i.e., of process $map(i)$), updating its local state and shared variables, or
\emph{context-switch} to a different active thread by changing $i$ to a different thread-index,
provided that in $T_i$ we are not in a block of sequential statements to be executed atomically.

\paragraph{\bf Verification under bounded round schedules:}

Fix a parameterized program ${\cal P}=(S,{\tt init},\{P_i\}_{i=1}^{n})$.
The verification problem asks, given a parameterized program ${\cal P}$,
whether every execution of the program respects all assertions.

%

In this paper, we consider a restricted verification problem.
A $k$-round schedule is a schedule that, for some ordering of the threads $T_1, \ldots, T_m$,
activates threads in $k$ rounds, where in each round, each thread is scheduled (for any number of steps)
according to this order. Note that an execution under a $k$-round schedule  (\emph{$k$-round execution}) 
can execute an unbounded number of steps.
%
Given a parameterized program and $k \in \mathbb{N}$, the verification problem for parameterized programs
under bounded round schedules asks whether any assertion is violated in some $k$-round execution.


\section{Linear interfaces}\label{sec:linearInterface}

\begin{figure}[tb]
\begin{minipage}{228pt}
{
\small
 \setlength{\unitlength}{0.70mm}
\begin{picture}(294,70)(15,-82)

\node[linegray=0.6,linewidth=0.5,Nw=125.0,Nh=80.0,Nmr=0.0](n280)(92.5,-50.8){}

\node[linewidth=0.4,Nw=15,Nh=70.0,Nmr=0.0](n24)(40,-48.05){}
\nodelabel[NLangle=-90.0,NLdist=39.0](n24){{$T_1$}}

\node[Nfill=y,fillcolor=Black,NLangle=0.0,NLdist=4,Nw=1.4,Nh=1.4,Nmr=1.0](n111)(40,-16.86){\small $s_1^1$}
\nodelabel[NLangle=180.0,NLdist=16.0](n111){{$u_1$}}
\node[Nfill=y,fillcolor=Black,NLangle=180.0,NLdist=4,Nw=1.4,Nh=1.4,Nmr=1.0](n112)(40,-26.05){\small $t_1^1$}
\node[Nfill=y,fillcolor=Black,NLangle=0.0,NLdist=4,Nw=1.4,Nh=1.4,Nmr=1.0](n113)(40,-33.05){\small $s_1^2$}
\nodelabel[NLangle=180.0,NLdist=16.0](n113){{$u_2$}}
\node[Nfill=y,fillcolor=Black,NLangle=180.0,NLdist=4,Nw=1.4,Nh=1.4,Nmr=1.0](n114)(40,-43.09){\small $t_1^2$}
\node[Nfill=y,fillcolor=Black,NLangle=0.0,NLdist=4,Nw=1.4,Nh=1.4,Nmr=1.0](n116)(40,-68.05){\small $s_1^k$}
\nodelabel[NLangle=180.0,NLdist=16.0](n116){{$u_k$}}
\node[Nfill=y,fillcolor=Black,NLangle=180.0,NLdist=4,Nw=1.4,Nh=1.4,Nmr=1.0](n117)(40,-78.05){\small $t_1^k$}

\drawline[dash={1.5}0](27,-16.86)(38,-16.86)
\drawline[dash={1.5}0](27,-33.05)(38,-33.05)
\drawline[dash={1.5}0](27,-68.05)(38,-68.05)

\drawedge(n111,n112){}
\drawedge[AHnb=0,curvedepth=4](n112,n113){}
\drawedge(n113,n114){}
\drawedge(n116,n117){}

\drawline[AHnb=0,dash={0.5 3.5}0](40,-45)(40,-63)


\drawline[AHnb=0,dash={0.5 3.5}0](23,-37)(23,-62)

\drawline[AHnb=0,dash={0.5 3.5}0](90,-37)(90,-62)

\drawline[AHnb=0,dash={0.5 3.5}0](120,-47)(120,-72)

\node[linewidth=0.4,Nw=15,Nh=70.0,Nmr=0.0](n80)(65,-48.05){}
\nodelabel[NLangle=-90.0,NLdist=39.0](n80){{$T_2$}}

\node[Nfill=y,fillcolor=Black,NLangle=0.0,NLdist=4,Nw=1.4,Nh=1.4,Nmr=1.0](n221)(65,-16.86){\small $s_2^1$}
\node[Nfill=y,fillcolor=Black,NLangle=180.0,NLdist=4,Nw=1.4,Nh=1.4,Nmr=1.0](n222)(65,-26.05){\small $t_2^1$}
\node[Nfill=y,fillcolor=Black,NLangle=0.0,NLdist=4,Nw=1.4,Nh=1.4,Nmr=1.0](n223)(65,-33.05){\small $s_2^2$}
\node[Nfill=y,fillcolor=Black,NLangle=180.0,NLdist=4,Nw=1.4,Nh=1.4,Nmr=1.0](n224)(65,-43.09){\small $t_2^2$}
\node[Nfill=y,fillcolor=Black,NLangle=0.0,NLdist=4,Nw=1.4,Nh=1.4,Nmr=1.0](n226)(65,-68.05){\small $s_2^k$}
\node[Nfill=y,fillcolor=Black,NLangle=180.0,NLdist=4,Nw=1.4,Nh=1.4,Nmr=1.0](n227)(65,-78.05){\small $t_2^k$}

\drawedge(n221,n222){}
\drawedge[AHnb=0,curvedepth=4](n222,n223){}
\drawedge(n223,n224){}
\drawedge(n226,n227){}

\drawedge[dash={1.5}0](n112,n221){}
\drawedge[dash={1.5}0](n114,n223){}
\drawedge[dash={1.5}0](n117,n226){}

\drawline[AHnb=0,dash={0.5 3.5}0](65,-45)(65,-63)


\node[Nfill=y,fillcolor=Black,NLangle=180.0,NLdist=5,Nw=1.4,Nh=1.4,Nmr=1.0](n331)(90,-16.86){$$}
\node[Nfill=y,fillcolor=Black,NLangle=180.0,NLdist=5,Nw=1.4,Nh=1.4,Nmr=1.0](n333)(90,-33.05){$$}
\node[Nfill=y,fillcolor=Black,NLangle=180.0,NLdist=5,Nw=1.4,Nh=1.4,Nmr=1.0](n336)(90,-68.05){$$}

\drawedge[dash={1.5}0](n222,n331){}
\drawedge[dash={1.5}0](n224,n333){}
\drawedge[dash={1.5}0](n227,n336){}

\node[linewidth=0.4,Nw=15,Nh=70.0,Nmr=0.0](n45)(145,-48.05){}
\nodelabel[NLangle=-90.0,NLdist=39.0](n45){{$T_m$}}

\node[Nfill=y,fillcolor=Black,NLangle=0.0,NLdist=4,Nw=1.4,Nh=1.4,Nmr=1.0](n551)(145,-16.86){\small $s_m^1$}
\node[Nfill=y,fillcolor=Black,NLangle=180.0,NLdist=4,Nw=1.4,Nh=1.4,Nmr=1.0](n552)(145,-26.05){\small $t_m^1$}
\nodelabel[NLangle=0.0,NLdist=18.0](n552){{$v_1$}}
\node[Nfill=y,fillcolor=Black,NLangle=0.0,NLdist=4,Nw=1.4,Nh=1.4,Nmr=1.0](n553)(145,-33.05){\small $s_m^2$}
\node[Nfill=y,fillcolor=Black,NLangle=180.0,NLdist=4,Nw=1.4,Nh=1.4,Nmr=1.0](n554)(145,-43.09){\small $t_m^2$}
\nodelabel[NLangle=0.0,NLdist=18.0](n554){{$v_2$}}

\node[Nfill=y,fillcolor=Black,NLangle=0.0,NLdist=4,Nw=1.4,Nh=1.4,Nmr=1.0](n556)(145,-68.05){\small $s_m^k$}
\node[Nfill=y,fillcolor=Black,NLangle=180.0,NLdist=4,Nw=1.4,Nh=1.4,Nmr=1.0](n557)(145,-78.05){\small $t_m^k$}
\nodelabel[NLangle=0.0,NLdist=18.0](n557){{$v_k$}}


\drawedge(n551,n552){}
\drawedge[AHnb=0,curvedepth=4](n552,n553){}
\drawedge(n553,n554){}
\drawedge(n556,n557){}

\node[Nfill=y,fillcolor=Black,NLangle=180.0,NLdist=5,Nw=1.4,Nh=1.4,Nmr=1.0](n442)(120,-26.05){}
\node[Nfill=y,fillcolor=Black,NLangle=180.0,NLdist=5,Nw=1.4,Nh=1.4,Nmr=1.0](n444)(120,-43.09){}
\node[Nfill=y,fillcolor=Black,NLangle=180.0,NLdist=5,Nw=1.4,Nh=1.4,Nmr=1.0](n447)(120,-78.05){}
\drawedge[dash={1.5}0](n442,n551){}
\drawedge[dash={1.5}0](n444,n553){}
\drawedge[dash={1.5}0](n447,n556){}

\drawline[dash={1.5}0](147,-26.05)(160,-26.05)
\drawline[dash={1.5}0](147,-43.09)(160,-43.09)
\drawline[dash={1.5}0](147,-78.05)(160,-78.05)

\drawline[AHnb=0,dash={0.5 3.5}0](145,-45)(145,-63)

\drawline[AHnb=0,dash={0.5 3.5}0](162,-47)(162,-72)



\drawline[AHnb=0,dash={0.5 3.5}0](95,-21.86)(115,-21.86)
\drawline[AHnb=0,dash={0.5 3.5}0](95,-38.05)(115,-38.05)
\drawline[AHnb=0,dash={0.5 3.5}0](95,-73.05)(115,-73.05)

\drawline[dash={1.5}0](177,-16.86)(188,-16.86)
\node[Nfill=y,fillcolor=Black,NLangle=0.0,NLdist=4,Nw=0.1,Nh=0.1,Nmr=1.0](n)(215,-16.86){\scriptsize{have same shared state}}

\drawline[AHnb=1](181,-31.86)(181,-41.86)
\node[Nfill=y,fillcolor=Black,NLangle=0.0,NLdist=4,Nw=0.1,Nh=0.1,Nmr=1.0](n)(215,-36.86){\scriptsize{local computation}}
\node[Nfill=y,fillcolor=White,NLangle=0.0,NLdist=4,Nw=0.1,Nh=0.1,Nmr=1.0](n)(215,-40.86){\scriptsize{(arbitrarily many events)}}

\node[Nfill=y,fillcolor=Black,NLangle=0.0,NLdist=4,Nw=0.1,Nh=0.1,Nmr=1.0](n1)(181,-55.86){}
\node[Nfill=y,fillcolor=Black,NLangle=0.0,NLdist=4,Nw=0.1,Nh=0.1,Nmr=1.0](n2)(181,-61.86){}
\drawedge[AHnb=0,curvedepth=4](n1,n2){}
\node[Nfill=y,fillcolor=Black,NLangle=0.0,NLdist=4,Nw=0.1,Nh=0.1,Nmr=1.0](n)(215,-58.86){\scriptsize{have same local state}}


\end{picture}
}
\end{minipage}

\caption{A linear interface}\label{fig:theoLI}
\end{figure}
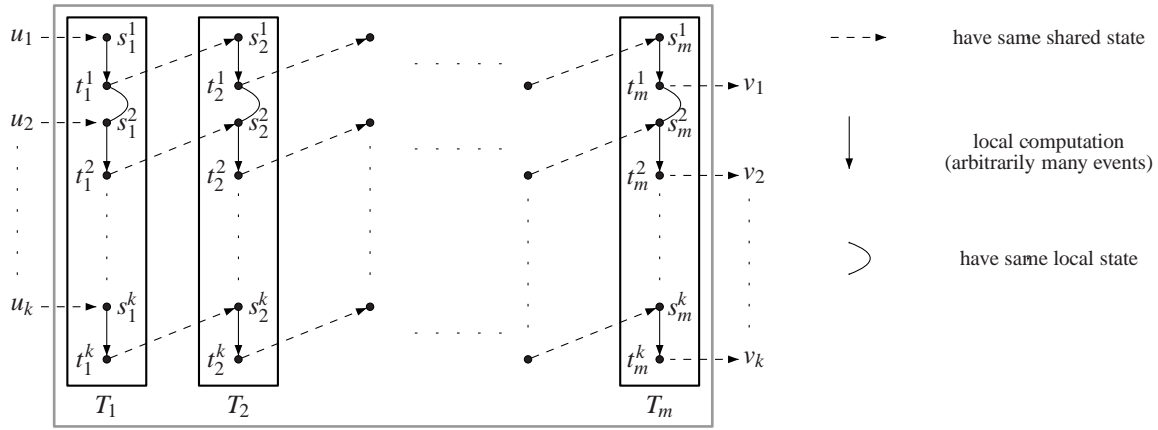

We now introduce the concept of a linear interface, which captures
the effect a block of processes has on the shared state, when involved in a $k$-round execution. 
 The notion of linear interfaces will play a major
role in the lazy conversion to sequential programs.

We fix a parameterized program ${\cal P}=(S,{\tt init},\{P_i\}_{i=1}^{n})$
and a bound $k>0$ on the number of rounds.

Notation: let $\overline{u}=(u_1,\ldots,u_k)$, where each $u_i$ is a shared state of ${\cal P}$.

A pair of $k$-tuples of shared variables $(\overline{u}, \overline{v})$ is
a \emph{linear interface} of length $k$ (see Figure~1) if:
{\bf (a)} there is an ordered block of threads $T_1,\ldots,T_m$ (running processes of ${\cal P}$),
{\bf (b)} there are $k$ rounds of execution, where each execution starts from shared state $u_i$,
exercises the threads in the block one by one, and ends with shared state $v_i$ (for example,
in Figure~1, the first round takes $u_1$ through states $s_1^1$, $t_1^1$, $s_2^1$, $t_2^1, \ldots$ to
$t_m^1$ where the shared state is $v_1$),
and {\bf (c)} the  local state of threads is preserved between consecutive rounds
 (in Figure~1, for example, $t_1^1$ and $s_1^2$ have the same local state).
Informally, a linear interface is the \emph{effect} a block of threads
can have on the shared state in a $k$-round execution, in that they transform
$\overline{u}$ to $\overline{v}$ across the block.

Formally, we have the following definition
(illustrated by Figure~\ref{fig:theoLI}).

\begin{definition}({\sc Linear interface})~\cite{CAV10}~\\
Let $\overline{u}=(u_1,\ldots,u_k)$ and
$\overline{v}=(v_1,\ldots,v_k)$ be tuples of $k$ shared states of a parameterized
program $\cal P$ (with processes $P$).\\
The pair $(\overline{u},\overline{v})$ is a \emph{linear interface} of $\cal P$ of length $k$
if there is some number of threads $m \in \mathbb{N}$,
an assignment of threads to processes $map:[1,m] \rightarrow P$ and
states $s_i^j=(map,i,x_i^j,\sigma_1^{i,j},\ldots,\sigma_m^{i,j})$ and
$t_i^j=(map,i,y_i^j,\gamma_1^{i,j},\ldots,\gamma_m^{i,j})$ of $\cal P$
for  $i\in [1,m]$ and $j \in [1,k]$, such that for each $i\in [1,m]$:
\begin{itemize}
\item $x_1^j = u_j$ and $y_m^j=v_j$, for each $j \in [1,k]$;
\item $t_i^j$ is reachable from $s_i^j$  using only local transitions of process $map(i)$, for each $j \in [1,k]$;
\item $\sigma_i^{i,1}$ is an initial local state for process $map(i)$;
\item $\sigma_i^{i,j+1}=\gamma_i^{i,j}$ for each $j \in [1,k-1]$ (local states are preserved across rounds);
\item $x_{i+1}^j= y_i^j$, except when $i=m$ (shared states are preserved between context-switches of a single round);
\item $(t_{i}^j, s_{i+1}^j)$, except when $i=m$, is a context-switch.\hfill\qed
\end{itemize}

\end{definition}



Note that  the above definition of a linear interface 
places no restriction on the relation between $v_j$ and $u_{j+1}$--- all that we require is that the block of threads
must take as input $\overline{u}$ and compute $\overline{v}$ in the $k$ rounds, preserving the
local configuration of threads between rounds.

A linear interface $(\overline{u}, \overline{v})$ of length $k$ is  \emph{wrapped}
if $v_i = u_{i+1}$ for each $i \in[1,k-1]$, 
and is  \emph{initial}
if $u_1$
 is an initial shared state of $\cal P$.

For a wrapped initial linear interface, from the definition of linear interfaces it follows that
the $k$ pieces of execution demanded in the definition
can be stitched together to get a complete execution of the parameterized program, that starts from
an initial state. We say that an execution \emph{conforms} to a particular linear interface if it meets
the condition demanded in the definition. 

\begin{lemma}\label{wrapped-lin-int}\cite{CAV10}~~
Let ${\cal P}$ be a parameterized program. An execution of ${\cal P}$ is under a $k$-round schedule
iff it conforms to some wrapped initial linear interface of ${\cal P}$ of length $k$.\hfill\qed
\end{lemma}

Hence to verify a program ${\cal P}$ under $k$-round schedules,
it suffices to check for failure of assertions along executions that conform to some wrapped initial
interface of length $k$.

\section{Sequentializing parameterized programs}\label{concToSeq}
In this section, we present a  sequentialization of parameterized
programs that preserves assertion satisfaction. Our translation is
``lazy''  in that the states reachable in the
resulting program correspond to reachable states of the parameterized program.
Thus, it preserves invariants across the translation: an invariant that
holds at a particular statement in the concurrent program will hold at the corresponding
statement in the sequential program.

A simpler \emph{eager} sequentialization scheme for parameterized programs that reduces reachability of error
states for parameterized programs but explores \emph{unreachable states} as well,
can be obtained by a simple adaptation of the translation from concurrent programs with finitely
many threads to sequential programs given in~\cite{LalReps}.
This scheme consists of simulating each thread till completion across \emph{all the rounds},
before switching to the next thread, and then, at the end, checking if the execution
of all the threads corresponds to an actual execution of the parameterized program.
Nondeterminism is used to guess the number of threads, the schedule, and the
shared state at the beginning of each round. However, this translation explores unreachable
states, and hence does not preserve assertions across the translation.

%
%

\paragraph{\bf Motivating laziness:}
A lazy translation that explores only localized states reachable by a parameterized program
has obvious advantages over an eager translation. 
For example, if we subject the sequential program to model-checking using state-space exploration, the
lazy sequentialization has fewer reachable states to explore.
The lazy sequentialization has another interesting consequence, namely that the sequential
program will not explore unreachable parts of the state-space where invariants of the parameterized
program get violated or where executing statements leads to system errors due to undefined semantics
(like division-by-zero, buffer-overflows, etc.), as illustrated by the following example.

\begin{figure}[tb]
\hspace*{50pt}
\framebox{
\begin{minipage}{340pt}
{\tt bool} $\mathit{blocked}:=T${\tt;}\\
{\tt int}~$x:=0,~y:=0${\tt ;}

\begin{multicols}{2}
{\tt process~$P_1$:}\\
\hspace*{0.3truecm}{\tt main(\,) } {\tt begin}\\
\hspace*{0.6truecm}{\tt while\,$(\mathit{blocked})$ do\\
\hspace*{0.9truecm}  skip;\\
\hspace*{0.6truecm}od}\\
\hspace*{0.6truecm}{\tt assert(y!=0);}\\
\hspace*{0.6truecm}{\tt $x:=x/y$;} \\
\columnbreak\hspace*{0.3truecm}{\tt end}

{\tt process~$P_2$:}\\
\hspace*{0.3truecm}{\tt main(\,) } {\tt begin}\\
\hspace*{0.6truecm}{\tt $x:=12$;}\\
\hspace*{0.6truecm}{\tt $y:=2$;}\\
\hspace*{0.6truecm}{\tt $\mathit{blocked}:=F$;~\emph{ //unblock $P_1$}}\\
\hspace*{0.3truecm}{\tt end}
\end{multicols}
\end{minipage}
}
\caption{Assertion not preserved by the eager sequentialization.}\label{fig:invariants}
\end{figure}

\begin{example}
Consider an execution of the parameterized program $\cal P$ from Figure~\ref{fig:invariants}.
The program involves only two threads: $T_1$ which executes $P_1$ and $T_2$ which executes $P_2$.
Observe that any execution of $T_1$ cycles on the while-loop until
$T_2$ sets \emph{blocked} to false. But before this, $T_2$ sets $y$ to $2$ and
hence the assertion $(y\neq 0)$ is true in $P_1$.
However, in an execution of the simpler eager sequentialization, 
we would
simulate $P_1$ for $k$ rounds and then simulate $P_2$ through $k$ rounds.
In order to simulate $P_1$, the eager translation would \emph{guess} non-deterministically
a $k$-tuple of shared variables $u_1,\ldots,u_k$. Consider an execution where $u_1$
assigns \emph{blocked} to be true, and $u_2$  assigns \emph{blocked} to false and $y$ to $0$.
The sequential program would, in its simulation of $P_1$ in the first round,
reach the while-loop, and would jump to the second round to simulate $P_1$ from $u_2$.
Note that the assertion condition would fail, and will be duly noted by the sequential
program. But if the assertion wasn't there, the sequentialization would execute the statement $x:=x/y$,
which would results in a ``division by zero'' exception.
In short, $(y\neq 0)$ is not an invariant for the statement $x:=x/y$ in the eager sequentialization.
The lazy translation presented in the next section avoids such scenarios.
\qed
\end{example}

\subsection{Lazy sequentialization}\label{lazyTrans}

Without loss of generality,
we fix a  parameterized program ${\cal P}=(S,{\tt init},\{P\})$ over one process.
Note that this is not a restriction, as we can always build $P$ so that it makes a non-deterministic choice
at the beginning, and decides to behave as one of a set of processes. 
We also replace functions with return values to {\tt void} functions that communicate the return value
to the caller using global (unshared) variables.
Finally, we fix a bound $k>0$ on the number of rounds.

We perform a lazy sequentialization of a parameterized program $\cal P$ by building a sequential
program that computes linear interfaces. More precisely, at the core of our construction is
a function {\tt linear\_int} that takes as input a set of valuations of shared variables
$\zug{u_1, \ldots, u_i, v_1, \ldots v_{i-1}}$ (for some $1 \leq i \leq k$) and \emph{computes} a shared valuation
$s$ such that $(\zug{u_1, \ldots, u_i}, \zug{v_1, \ldots, v_{i-1}, s})$ is a linear interface.
We outline how this procedure works below.

\begin{sloppypar}
The procedure {\tt linear\_int} will require the following pre-condition, and meet the following
post-condition and invariant when called with the input $\zug{u_1, \ldots, u_i, v_1, \ldots v_{i-1}}$:
\begin{description}
 \item[Precondition:] There is some $v_0$, an initial shared state,
 such that       $(\zug{v_0, v_1, \ldots v_{i-1}}, \zug{u_1, u_2, \ldots u_i})$ is a linear interface.
 \item[Postcondition:] The value of the shared state at the return, $s$, is such that
       $(\zug{u_1, \ldots, u_i}, \zug{v_1, \ldots v_{i-1}, s})$ is a linear interface.
 \item[Invariant:] At any point in the execution of {\tt linear\_int}, if the localized state is
      $(\widehat{\sigma}, s)$, and a statement of the parameterized program is executed from
      this state, then $(\widehat{\sigma}, s)$ is a localized state reached in some
      execution of ${\cal P}$.
\end{description}
\end{sloppypar}

Intuitively, the pre-condition says that there must be a ``left'' block of threads where the initial
computation can start, and which has a linear interface of the above kind.
This ensures that all the $u_i$'s are indeed reachable in some computation.
Our goal is to build {\tt linear\_int} to sequentially compute, using nondeterminism,
any possible value of $s$ such that $(\zug{u_1, \ldots, u_i}, \zug{v_1, \ldots v_{i-1}, s})$ is a linear interface
(as captured by the post-condition). The invariant above assures laziness; recall that the laziness
property says that no statement of the parameterized program will be executed on a localized
state of the sequential program that is unreachable in the parameterized program.

Let us now sketch how {\tt linear\_int} works on an input $\zug{u_1, \ldots, u_i, v_1, \ldots v_{i-1}}$.
First, it will decide non-deterministically whether the linear interface
is for a \emph{single thread} (by setting a variable $\mathit{last}$ to $T$, signifying it is simulating the last thread)
or whether the linear interface is for a block of threads more than one (in which case $\mathit{last}$ is
set to $F$).

It will start with the state $(\sigma_1,u_1)$ where $\sigma_1$ is an initial local
state of $P$, and simulate an arbitrary number of moves of $P$, and stop this simulation at some
point, non-deterministically, ending in state $(\sigma_1',u_1')$. At this point, we would like the
computation to ``jump'' to state $(\sigma_1', u_2)$, 
however we need first to ensure that this state is reachable.

If $\mathit{last}=T$, i.e. if the thread we are simulating is the last thread, then this is easy,
as we can simply check if $u_1' = v_1$.
If $\mathit{last}=F$, then {\tt linear\_int} determines whether $(\sigma_1',u_2)$ is reachable
by \emph{calling itself recursively} on the tuple $\zug{u_1'}$, getting the return value $s$,
and checking whether $s=v_1$. In other words, we claim that
$(\sigma_1', u_2)$ is reachable in the parameterized program if $(u_1',v_1)$ is a linear interface.

Here is a proof sketch. Assume $(u_1',v_1)$ is a linear interface; then by the pre-condition we know that there
is an execution starting from a shared initial state to the shared state $u_1$. By switching to
the current thread $T_h$ and using
the local computation of process $P$ just witnessed, we can take the state to $u_1'$ (with local
state $\sigma_1'$), and since $(u_1',v_1)$ is a linear interface, we know there is a ``right'' block of
processes that will somehow take us from $u_1'$ to $v_1$. Again by the pre-condition, we know that
we can continue the computation in the second round, and ensure that the state reaches $u_2$,
at which point we switch to the current thread $T_h$, to get to the local state $(\sigma_1',u_2)$.

The above argument is the crux of the idea behind our construction. In general, when we have
reached a local state $(\sigma_i',u_i')$, {\tt linear\_int} will call itself on the tuple
$u_1', \ldots, u_i', v_1, \ldots v_{i-1}$, get the return value $s$ and check if $s=v_i$, before it
``jumps'' to the state $(\sigma_i', u_{i+1})$. Note that when it calls itself, it maintains
the pre-condition that there is a $v_0$ such that $(\zug{v_0, v_1, \ldots v_{i-1}}, \zug{u_1', \ldots, u_i'})$
is a linear interface by virtue of the fact that the pre-condition to the current call holds,
and by the fact that the values $u_1', \ldots, u_i'$ were computed consistently in the current
thread.

The soundness of our construction depends on the above argument. Notice that the laziness invariant
is maintained because the procedure calls itself to check if there
is a ``right'' block whose linear interface will witness reachability, and the computation involved in this is assured
to meet only reachable states because of its pre-condition which demands that there is a ``left''-block
that assures an execution.

Completeness of the sequentialization relies on several other properties of the construction. First,
we require that a call to {\tt linear\_int} returns \emph{all possible values} of $s$ such that $(\zug{u_1, \ldots, u_i}, \zug{v_1, \ldots v_{i-1}, s})$ is a linear interface. Moreover, we need that for every execution
corresponding to this linear interface and every local state $(\sigma,s)$ seen along such an
execution, the local state is met along some run of {\tt linear\_int}. It is not hard
to see that the above sketch of {\tt linear\_int} does meet these requirements.

Notice that when simulating a particular thread, two successive calls to {\tt linear\_int} may result in
different depths of recursive calls to {\tt linear\_int}, which means that a different number of
threads are explored. However, the correctness of the computation does not depend on this,
as correctness only relies on the fact that {\tt linear\_int} computes a  linear
interface, and the number of threads in the block that witnesses this interface is immaterial.
This property of a linear interface that encapsulates a block of threads no matter how their
internal composition is, is what makes a sequentialization without extra counters possible.

We will have a {\tt main} function that drives calls to {\tt linear\_int}, calling it to compute
linear interfaces starting from a shared state $u_1$ that is an initial shared state.
Using successive calls, it will construct linear interfaces of the form $\zug{u_1, \ldots, u_i, v_1, \ldots, v_i}$ maintaining that $v_j=u_{j+1}$, for each $j < i$. This will ensure that the interfaces it computes
are \emph{wrapped} interfaces, and hence the calls to {\tt linear\_int} meet the latter's pre-condition.
When it has computed a complete linear interface of length $k$, it will stop, as any localized state
reachable in a $k$-round schedule would have been seen by then (see Lemma~\ref{wrapped-lin-int}).

\paragraph{\bf The syntactic transformation.}
\begin{figure}[tb]

\hspace*{30pt}
\framebox{
\begin{minipage}{380pt}
\small
\emph{Denote} $\overline{q}_{i,j}=q_i,\ldots,q_j$\\
\emph{Let s be the shared variables and g the global variables of $\cal P$;  }\\
\hspace*{0.1truecm}{\tt bool} $\mathit{\atomic, terminate}${\tt ;}\\

\vspace*{-0.3truecm}
\begin{multicols}{2}
{\tt main(\,) }\\
{\tt begin}\\
\hspace*{0.3truecm}{\tt Let $q_1,\ldots,q_k$ be of type of $s$;}\\
\hspace*{0.3truecm}{\tt int $i=1$;}\\
\hspace*{0.3truecm}$\mathit{\atomic}:= F${\tt ;}\\
\hspace*{0.3truecm}{\tt call init();}\\
\hspace*{0.3truecm}$q_1:=g${\tt ;}\\
\hspace*{0.3truecm}{\tt while\,$(i \le k)$ do}\\
\hspace*{0.7truecm}$\mathit{terminate}:=F${\tt ;}\\
\hspace*{0.7truecm}{\tt call linear\_int$(\overline{q}_{1,k},\overline{q}_{2,k},i)$;}\\
\hspace*{0.7truecm}$i${\tt ++;}\\
\hspace*{0.7truecm}{\tt if\,$(i\le k)$ then}\\
\hspace*{1.1truecm}$q_i:=s${\tt ;}\\
\hspace*{0.7truecm}{\tt fi}\\
\hspace*{0.3truecm}{\tt od}\\
\hspace*{0.3truecm}{\tt return;}\\
{\tt end}\\
\\
\\
\\
\underline{\emph{Interlined code}}:\\
\hspace*{0.1truecm}{\tt if\,($\mathit{terminate}$) then return; fi}\\
\hspace*{0.1truecm}{\tt if\,($\neg\mathit{\atomic}$) then}\\
\hspace*{0.5truecm}{\tt while\,(*) do}\\
\hspace*{0.9truecm}{\tt if\,($\mathit{last}$) then}\\
\hspace*{1.3truecm}{\tt if\,($j=\mathit{bound}$) then}\\
\hspace*{1.7truecm}{\tt $\mathit{terminate}:=T$;~return;}\\
\hspace*{1.3truecm}{\tt else}
~{\tt assume($q'_j=s$);}\\
\hspace*{1.7truecm}{\tt $j$++; ~$s:=q_j$;}\\ 
\hspace*{1.3truecm}{\tt fi}\\
\hspace*{0.9truecm}{\tt else}\\
\hspace*{1.3truecm}$q_j:=s${\tt ;} ~\emph{save}$:=g${\tt ;} \\
\hspace*{1.3truecm}{\tt call linear\_int$(\overline{q}_{1,k},\overline{q}'_{1,k-1},j)$;}\\
\hspace*{1.3truecm}{\tt if\,($j=\mathit{bound}$) then}
~~~{\tt return;}\\
\hspace*{1.3truecm}{\tt else}
~~~{\tt assume($q'_j=s$); }~$g:=$\emph{save}{\tt ;}\\
\hspace*{1.7truecm} {\tt $\mathit{terminate}:=F$;\,$j$++;\,$s:=q_j$;}\\ 
\hspace*{1.3truecm}{\tt fi}\\
\hspace*{0.9truecm}{\tt fi}\\
\hspace*{0.5truecm}{\tt od}\\
\hspace*{0.1truecm}{\tt fi}
\end{multicols}
\end{minipage}
}
\caption{Function {\tt main} and interlined control code of the sequential program
${\cal P}^{\mathit{lazy}}_k$.}\label{fig:lazy}
\vspace*{-0.5truecm}
\end{figure}

The sequential program ${\cal P}^{\mathit{lazy}}_k$ obtained  from ${\cal P}$ in the lazy sequentialization
consists of the function {\tt init} of $\cal P$, a new function {\tt main},
a function ${\tt linear\_int}$,  and for every function $f$ other than
{\tt main} in $\cal P$, a function $f^{\mathit{lazy}}$. The  function  {\tt main} of
 ${\cal P}^{\mathit{lazy}}_k$ is shown in Figure~\ref{fig:lazy}. The function ${\tt linear\_int}$
is obtained by transforming the {\tt main} function in the process of the
parameterized program, by interlining the code  shown in Figure~\ref{fig:lazy} between every two statements.
 Each  functions $f^{\mathit{lazy}}$ 
is obtained from $f$ similarly by inserting the same interlined code. Clearly, in these transformations 
each call to $f$ gets replaced with a call to $f^{\mathit{lazy}}$.

The interlined code allows to interrupt the simulation of a thread (provided we are not in an atomic
section), and either jump directly to the next shared state (if $\mathit{last}=1$) or call recursively
{\tt linear\_int} to ensure that jumping to the next shared state will explore a reachable state.
Observe that
before calling ${\tt linear\_int}$ recursively from
the interlined code, we copy $g$ (i.e., the value of $P$'s global
variables)
to the local variables \emph{save}, and after returning, copy it back to $g$
to restore the local state.

The global variables of ${\cal P}^{\mathit{lazy}}_k$ includes all global and shared variables of $\cal P$,
as well as two extra global Boolean variables \emph{\atomic} and \emph{terminate}.
The variable \emph{\atomic} is used to flag that the simulation is within an atomic block of instructions
where context-switches are prohibited. The variable \emph{terminate} is used to force the return from
the most recent call to ${\tt linear\_int}$ in the call stack
(thus all the function calls which are in the call stack up to
this call are also returned).
This variable is set false in the beginning and after returning each call to
${\tt linear\_int}$.

Function {\tt main} uses $k$ copies of the shared variables
denoted with $q_1,\ldots,q_k$.
It calls {\tt init} and then iteratively calls
${\tt linear\_int}$ with $i=1,\ldots,k$.
Variable $q_1$ is assigned in the beginning and
at each iteration $i < k$ the value of the shared variables is stored in $q_{i+1}$.

Function ${\tt linear\_int}$ is defined with formal parameters
$\overline{q}=\tuple{q_1,\ldots,q_k}$,
$\overline{q}'=\tuple{q'_1,\ldots,q'_{k-1}}$ and $\mathit{bound}$.
Variable \emph{bound} stores the bound on the number
of rounds to execute in the current call to ${\tt linear\_int}$.

The variable \emph{\atomic} is set to true when entering an atomic block and set back to false
on exiting it. The interlined code refers to
variables $\mathit{last}$ and $j$. The variable $\mathit{last}$  is nondeterministically
assigned when ${\tt linear\_int}$ starts.
Variable $j$ counts the rounds being executed so far in the current
call of ${\tt linear\_int}$ ($j$ is initialized to $1$).

We also insert
``{\tt assume}(F);" before each return statement of
${\tt linear\_int}$ which is not part of the interlined code;
this prevents a call to ${\tt linear\_int}$ to be returned
after executing to completion.


\ignore{
In the rest of this section we show some relevant properties
of ${\tt process}^{\mathit{lazy}}$ using induction on the structure of the recursive calls
to ${\tt process}^{\mathit{lazy}}$
itself. These properties will be used in turn to show that the resulting
program ${\cal P}^{\mathit{lazy}}_k$ is indeed ``reachability equivalent'' to
$\cal P$ up to $k$ rounds, and that along its executions the exploration of  $\cal P$ states
happens lazily.
}

\paragraph{\bf Correctness and laziness of the sequentialization}~\\
We now formally prove the correctness and laziness of our sequentialization.
%
We start with a lemma stating that function ${\tt linear\_int}$ indeed
computes linear interfaces of the parameterized program $\cal P$
(i.e. meets its post-condition).
\begin{lemma}\label{lemma:post-condition}
Assume that ${\tt linear\_int}$ when called with actual parameters
$u_1,\ldots,u_k,$ $v_1,\ldots,v_{k-1},i$ terminates and returns.
If $\widehat{s}$ is the valuation of the (global) variable $s$ at return,
then $(\tuple{u_1,\ldots,u_i},\tuple{v_1,\ldots,v_{i-1},\widehat{s}})$
is a linear interface of {\cal P}.\hfill\qed
\end{lemma}

Consider a call to ${\tt linear\_int}$ such that the precondition stated
in page~\pageref{lazyTrans} holds.
Using the above lemma we can show that the localized states from which
we simulate the transition of $\cal P$ are discovered lazily, and
that the program ensures that the precondition holds on recursive calls to ${\tt linear\_int}$.

\begin{lemma}\label{lemma:preserveInvariant}
Let $(\tuple{v_0, v_1,\ldots,v_{i-1}},\tuple{u_1,\ldots,u_i})$
be an initial linear interface.
Consider a call
to ${\tt linear\_int}$ with actual parameters
$u_1,\ldots,u_k,$ $v_1,\ldots,v_{k-1},i$.
\vspace*{-0.1cm}
\begin{itemize}
 \item Consider a localized state reached during an execution of this call, and
       let a statement of $\cal P$ be simulated on this state.
       Then the localized state is reachable in some execution of {\cal P}.
 \item Consider a recursive call to {\tt linear\_int} with parameters $u'_1,\ldots,u'_k,$ $v_1,\ldots,v_{k-1},j$. Then \linebreak
$(\tuple{v_0, v_1,\ldots,v_{j-1}},\tuple{u_j',\ldots,u_j'})$ is a linear
       interface.\hfill\qed
\end{itemize}
\end{lemma}

Note that whenever the function {\tt main} calls {\tt linear\_int}, it satisfies the pre-condition
for {\tt linear\_int}. This fact along with the above two lemmas establish the soundness
and laziness of the sequentialization.

The following lemma captures the completeness argument:
\begin{lemma}\label{lemma:completeness}
Let $\rho$ be a $k$-round execution of $\cal P$.
Then there is a wrapped initial linear interface \linebreak
$(\tuple{u_1,\ldots,u_k},\tuple{u_2,\ldots,u_k,v})$
that $\rho$ conforms to,
and a terminating execution $\rho'$ of ${\cal P}^{\mathit{lazy}}_k$ such that at the end of  $\rho'$,
the valuation of the variables $\zug{q_1, \ldots, q_k, s}$ is $\zug{u_1, \ldots, u_k, v}$.
Furthermore, every localized state reached in $\rho$ is also reached in $\rho'$.\hfill\qed
%
\end{lemma}

\noindent Consolidating the above lemmas, we have:
\begin{theorem}\label{theo:correct-complete-lazy}
Given $k \in \mathbb{N}$ and a parameterized program $\cal P$,
an assertion is violated in a $k$-round  execution of $\cal P$
if and only if an assertion is violated in
an execution of ${\cal P}^{\mathit{lazy}}_k$.
Moreover, ${\cal P}^{\mathit{lazy}}_k$ is lazy:
if  ${\cal P}^{\mathit{lazy}}_k$ simulates a statement of $\cal P$
on a localized state, then the localized state is reachable in $\cal P$.\hfill\qed
\end{theorem}

\subsection{\bf Parameterized programs over finite data domains:}
A sequential program with variables ranging over finite domains can be modeled as a
pushdown system. Analogously, a parameterized program with variables ranging over finite
domains can be modeled as a parameterized multi-stack pushdown system,
i.e., a  system composed of a finite number of pushdown systems sharing a portion of the
control locations, which can be replicated in an arbitrary number of copies in each run.
A \emph{parameterized multi-stack pushdown system} ${\cal A}$ is thus a
tuple $(S,S_0,\{A_i\}_{i=1}^n)$, where $S$ is a finite set of shared locations, $S_0\subseteq S$ is
the set of the initial shared locations and  for $i\in[1,n]$, with
$A_i$ is a standard pushdown system whose set of control locations is $S\times L_i$ for some finite set $L_i$.
We omit a formal definition of the behaviors of  $\cal A$ which can be
easily derived from the semantics of parameterized programs given in Section \ref{sec:pre},
by considering that each $s\in S$ is the analogous of
 a shared state in the parameterized programs, a state of each $A_i$ is the analogous of
a local state of a process, and thus $(s,l)\in S\times L_i$ corresponds to a localized state.

Following the sequentialization construction given earlier in this section
to construct the sequential program ${\cal P}_k^{\mathit{lazy}}$ from a parameterized
program $\cal P$, we can construct from $\cal A$ a pushdown system ${\cal A}_k$ such that
the reachability problem under $k$-round schedules in $\cal A$ can be reduced to
the standard reachability problem in ${\cal A}_k$. Also, the number of locations of ${\cal A}_k$ is
$O(\ell \,k^2\, |S|^{2k})$ and the number of transitions of ${\cal A}_k$ is
$O(\ell\,d\,k^3\,|S|^{2k-1})$ where $\ell$ is $\sum_{i=1}^n |L_i|$ and
$d$ is the number of the transitions of $A_1,\ldots,A_n$.

\begin{theorem}\label{th:PDS}
Let $\cal A$ be a parameterized multi-stack pushdown system and $k\in \mathbb{N}$.
Reachability up to $k$-round schedules in $\cal A$ reduces to reachability
in ${\cal A}_k$. Moreover, the size of ${\cal A}_k$ is singly exponential in $k$ and
linear in the product of the number of locations and transitions of  $\cal A$.
\hfill\qed
\end{theorem}

\section{Conclusions and Future Work}
We have given an assertion-preserving
efficient sequentialization of parameterized concurrent programs under bounded round schedules.
%

An interesting future direction is to practically utilizing the sequentialization
to analyze parameterized concurrent programs. 
For concurrent programs with a finite number of threads, \emph{bounded-depth}
verification using SMT solvers has worked well, especially using eager translations~\cite{ZvonimirCAV09,ZvonimirSPIN}.
However, since the sequentialization described in this paper introduces recursion
even for bounded-depth concurrent programs, it would be hard to verify the resulting
sequential program using SMT solvers.
We believe that verifying the sequential program using abstract interpretation techniques
that are context-sensitive would be an interesting future direction to pursue; in this context, 
the laziness of the translation presented here would help in maintaining the accuracy of the analysis.

%



Finally, sequentializations can also be used to subject parameterized programs to
abstraction-based model-checking. It would be worthwhile to pursue under-approximation of
static analysis of concurrent and parameterized programs (including data-flow and
points-to analysis) using sequentializations.



\end{document}